\documentclass[a4paper, 12pt ,onecolumn]{IEEEtran}
\usepackage[nolists,tablesfirst,nomarkers]{endfloat} %figuras no final
\usepackage{times, epsfig}
\usepackage{amssymb}
\usepackage{amsmath}
\usepackage{graphicx}
\usepackage[latin1]{inputenc}
\usepackage{graphicx}
\usepackage{setspace}
\usepackage{psfrag}
\usepackage{cite}
\usepackage{multirow}
\usepackage{wasysym}
\usepackage{multirow}
\usepackage{float}
\usepackage{pdfsync}
\usepackage{pstricks}
\newtheorem{theorem}{Theorem}

\newtheorem{claim}{Claim}
\newtheorem{proof}{Proof}
\newtheorem{definition}{Definition}

\begin{document}

\title{Interference Networks: A Complex System View}
\author{Pedro H. J. Nardelli, \IEEEmembership{Student Member, IEEE}, Paulo Cardieri, \IEEEmembership{Member, IEEE} \\ William A. Kretzschmar Jr. and Matti Latva-aho, \IEEEmembership{Senior Member, IEEE}
\thanks{This work was supported in part by Infotech Graduate School,  Finnish Funding Agency for Technology, and Innovation (Tekes), Renesas Mobile, Nokia Siemens Networks, Elektrobit and CNPq-Brazil (Proc. No. 309597/2009-9).

Pedro H. J. Nardelli is with Center for Wireless Communications (CWC), University of Oulu, Finland, and Wireless Technology Laboratory (WissTek), State University of Campinas, Brazil (e-mail:~nardelli@ee.oulu.fi).
Paulo Cardieri is with the WissTek.
William A. Kretzschmar Jr. is with Dept. of English at University of Georgia, USA, and Faculty of Humanities at University of Oulu.
Matti Latva-aho is with CWC.
}}
\maketitle
%\markboth{Submission v1: \today}{}

%****************************************************************************
%****************************************************************************
\begin{abstract}
This paper presents an unusual view of interference wireless networks based on complex system thinking.
To proceed with this analysis, a literature review of the different applications of complex systems is firstly presented to illustrate how such an approach can be used in a wide range of research topics, from economics to linguistics.
Then the problem of quantifying the fundamental limits of wireless systems where the co-channel interference is the main limiting factor is described and hence contextualized in the perspective of complex systems.
Specifically some possible internal and external pressures that the network elements may suffer are identified as, for example, queue stability, maximum packet loss rate and transmit power constraint.
Besides, other important external factors such as mobility and incoming traffic are also pointed out.
As a study case, a decentralized point-to-point interference network is described and several claims about the optimal design setting for different network states and under two mobility conditions, namely quasi-static and highly mobile, are stated based on results found in the literature.
Using these claims as a background, the design of a robust adaptive algorithm that each network element should  run  is  investigated.
\end{abstract}

\begin{keywords}
Adaptive algorithms, complex systems,  decentralization,  interference networks, self-organization
\end{keywords}

%****************************************************************************
%****************************************************************************
\section {Introduction}
\label{Sec_Intro}
Complexity is a term used in several diverse research fields, from theoretic physics to social sciences or ecology, to characterize a state that is neither completely deterministic nor completely random that \emph{emerges} from the dynamics of systems whose different elements interact amongst themselves and may adapt their relation rules in accordance to both internal and external factors \cite{Mitchell2009}.
As an example of how very simple rules may lead to unexpected intricate patterns over time, we can cite the extensive work about one-dimensional cellular automata done by Wolfram in \cite{Wolfram2002}.
There, the author plots the state of one-dimensional automata at each time period considered, creating then a two-dimensional figure that can be categorized in the following classes: stable, random, periodic and complex.
Without going into further details, what Wolfram's work tells us is that, even using a very simple spatial interrelation rule, systems \emph{may} generate complexity by \emph{themselves}.   
As this fact suggests, decentralized systems might be still functional even without any controlling entity, indicating that they are able to \emph{self-organize}.
Many illustrations of this can be found in nature as, for instance, ants working in colonies, neurons building a capable brain and so on \cite{Mitchell2009}. 
Different from these solutions offered by the \emph{invisible hand} of nature, however,  engineering systems designed by humans almost always do not have a long time to evolve self-organization in a proper way and also do not accept outputs without a minimum quality requirement.
For this reason, designed solutions that can both produce and sustain complex behaviors are usually a very hard task.
In addition, the conditions of controllability in complex network are still a new research field, which starts developing very recently \cite{Liu2011}.
Yet the development of practical engineering systems that allow for adaptation and self-organization is constantly increasing.

In this work, we focus our analysis on the fundamental limits of wireless, interference-limited networks from a complex-system thinking.
Our goal here is to identify the features  that the network elements should have for a given broad class of scenarios (i.e. static or mobile, dense or sparsely populated, etc.).
Once the capabilities of these elements are determined, our target is to study how their parameters should be \emph{adaptively} set as a function of the \emph{locally processed} external and internal state information in order to maintain the network functional or even optimize its efficiency, while its elements are still satisfying their own quality requirements.

Before we go deeper, it is worth mentioning that extensive work has been done to either characterize the capacity of communication networks  using information theory  \cite{ElGamal_2012} or design practical algorithms with self-organization capabilities for cellular networks \cite{Aliu2012}; however, very few have been done to combine both approaches.
Our aim here is to present how complexity science can put some light on the interference network problem, providing a more systemic view to it, while we maintain the rigor required by information theorists.
It is worth saying that this new way of seeing science is also changing how researchers from different fields cope with several problems where the traditional reductionist approach is not able to provide satisfactory answers as discussed in \cite{Barabasi2012}.

Next we present an informal description of the scenario under analysis, which will help the reader to get a better understanding of the problem we are dealing with.
%

%****************************************************************************
\subsection{Informal statement: chatting in a party problem}
\label{SubSec_Informal}
In this subsection we introduce the problem of people talking at a party as an illustration of the interference network that we will work on later.
Our goal with this informal statement is to provide some intuition on the problem and show how we, \emph{intelligent humans}, attempt to cope with it, imagining some possible decisions and their effects on the network.
For a clear parallel, we explicit whenever we believe appropriate which communication engineering aspect is mimic by our informal statement.

Let us consider that a couple goes to a party.
When they arrive, there are only few people around talking to each other.
In this situation, our reference couple can successfully have a chat; the others are talking in an acceptable intensity (transmit power is limited); they are most probably far from each other (random spatially distributed) and the background music (noise) is the main limiting factor of the conversation (noise-limited scenario).
After an hour, however, more people have arrived and thus more people are chatting, increasing the \emph{interference} level throughout the party place.
Moreover, persons are getting closer and closer to our reference couple, which consequently starts facing problems to communicate. 
The others start suffering the same problem as well.
What should each person do to improve his/her own performance that is affected by external factors? 
If everyone does the same, is the network still functional or, in other words, are people able to chat?

A straightforward decision when the interference from people in concurrent conversations is disturbing the couple under analysis is to speak louder (power control). 
This is in fact an optimal solution for a single couple; yet, everyone speaking louder harms the performance of the network and in the end of the day it is completely useless.
This is easy to visualize in parties and restaurants when everyone is screaming when talking.
So, even if it is optimal for one pair for a given fixed condition of the network, this is not a good decision for the network as whole because other people will also take the same decision of speaking louder. 
What else can be done then?

Another possible solution is to provide feedback regarding the success of what has been said (Automatic Repeat-reQuest, ARQ, protocol); if the message was not clearly understood by the listener, he would inform the other who will repeat whatever he said before.
This would work, but allowing for many repetitions would be inefficient since a successful communication might require many trials.
Even worse, if the speaker has a lot to say (arrival process) and stays repeating the same thing for long periods, he would probably forget something (buffer overflow, unstable queues).

Another possibility is to control when people talk (Multiple Access Control, MAC, protocols).
For example, a synchronization feature in the network could help so that people have time-slots to communicate.
Our reference couple can decide randomly whether to communicate in the beginning of a given slot (slotted-ALOHA MAC protocol).
Or in the absence of synchronization a random waiting time can be considered such that the person only starts her transmission after some random period (unslotted-ALOHA MAC protocol).
Other option is to wait for the interference level be in a satisfactory level to then start chatting (Carrier Sense Multiple Access, CSMA, based MAC protocol).
These mechanisms can help the communication, but it still faces issues related to long periods before starting a transmission, which harms the efficiency of the chatting.
As in the repetition strategy, if a person has many things to tell and it is very difficult to communicate successfully, she will probably start forgetting what she wanted to say.
A different strategy that could be used by our reference speaker is to say things slowly (lower coding rates).
This increases the chances of a successful understanding by her listener while it does not affect the others' chatting.
A drawback of this is that when many people are talking at the same time, the speaker should talk very slow, which in turn negatively affect  their communication efficiency.
However, if the listener tries also to understand some of the strongest interfering speakers that are chatting at the same time (joint detection), her actual speaker might talk a bit faster.

We just listed some possible strategies that help people chat in a party when the number of concurrent conversations increases.
As we argued, none of them alone can provide a successful answer for network variations, which also depends on internal pressures such as personal limitations on forgetting things before saying or losing information after some communication attempts.
Nevertheless, if we think how humans react in a real situation, they use \emph{combinations} of the possible strategies and they are normally able to talk.
Which strategies are employed and how to apply them are abilities \emph{learned} and \emph{acquired} from previous \emph{experiences}.
Furthermore, to be able to do so, every single person should have knowledge from the relevant information around them.
Therefore, the \emph{sensing} and \emph{processing} capabilities of humans are necessary conditions to successful chat in a party when people arrive and leave all the time.

In this paper, we mathematically characterize examples of this problem using concepts from different fields as follows: (i) stochastic geometry to determine the spatial distributions of the network, (ii) Shannon information theory to assess data rates and possible decoders, (iii) communication network theory to study access control and retransmissions protocols, (iv) signal processing to identify and process the relevant information of the network state, (v) game theory to compare the selfish and collective optimal solutions, (vi) queuing theory to evaluate the internal pressure related to the stability of the solution, (vii) theoretical biology to have a better understanding of how a good solution evolves from internal and external pressures and (viii) social sciences and institutional design to provide a view of how rules and incentives should be created to achieve a given desired behavior.

The rest of this paper is divided as follows.
In Section \ref{Sec_Related}, we highlight  the main related works concerning the capacity of ad hoc wireless networks and a literature review of some multidisciplinary problems that have a similar appeal and have been analyzed with complex science.
In Section \ref{sec_general_description}, we present a general description of the scenario studied here, including some possible strategies and evaluation metrics.
In Section \ref{sec_complview} we state our claims about more general guidelines and strategies that should be taking into account for different classes of interference networks (e.g. densely or sparsely populated and quasi-static or mobile network elements) and then we discuss how they can be used to design an adaptive distributed algorithm that is robust against variations.
Section \ref{sec_Concl} concludes this paper, indicating the road-map for future works.
%

%****************************************************************************
%****************************************************************************
\section {Related Works}
\label{Sec_Related}
%****************************************************************************
\subsection {Capacity of ad hoc wireless networks}
\label{SubSec_ReviewCapacity}
One of the first attempts to deal with a class of communication channels in which parallel transmissions occur dates back to the late 70's, when Carleial defined from an information-theoretic perspective the so-called \emph{interference channel} \cite{Carleial1978}.
Even though different bounds and strategies have been proposed to better characterize this class of channel, the capacity region of the simplest two-source-two-destination scenarios is still undetermined \cite[Ch. 6]{ElGamal_2012}.
To make things worse, larger networks where multiple sources and destinations coexist have their capacity regions even more unclear.
This fact has stimulated researchers to think about different ways to understand the limits of interference networks composed by several elements affecting each other communication links.
For example, we can cite these two magazine publications \cite{andrews_rethinking_2008,Goldsmith2011} to illustrate some possible research directions.
Another important contribution to characterize the capacity of wireless networks was introduced by Gupta and Kumar in their  seminal paper \cite{Gupta2000}.
In that work, the authors defined the \emph{transport capacity} metric to quantify how many bits-meter a wireless network can reliably sustain when the density of concurrent transmissions grows to infinite (asymptotic analysis).
Following it, many other studies have focused on establishing the capacity scaling laws for different scenarios and under different assumptions (e.g. \cite{Xue2006}).
Using an unconventional  perspective, Franceschetti et al. \cite{Franceschetti2009,Franceschetti2011} derived some fundamental properties of wireless networks relying on established methods of electrodynamic and electromagnetic theories.
Yet, the aforementioned works are strongly based on scenarios where the number of nodes in the network infinitely increases, which might misguide the design of actual strategies for either physical or medium access control network layers.
To cope with this limitation, the \emph{transmission capacity} was proposed by Weber et al. in \cite{WebYan2005} to statistically assess the highest average \emph{spatial throughput}\footnote{Spatial throughput can be also referred to as \emph {area spectral efficiency} \cite{Alouini1999}.} that the network can reach such that the link outage probability is bounded by a given small value and the density of active links is the optimization variable.
Since this first paper published in 2005, different strategies such as interference cancellation, threshold transmissions, guard zones, bandwidth partitioning amongst others have been analyzed using this framework and the main results have been recently compiled in \cite{Weber2012}.
Other important extensions following this line can be found in \cite{Vaze2011,KrishnaGanti2011,Nardelli2011,Nardelli2012}
Here it is worth stressing that this result makes use of stochastic geometry and spatial point processes \cite{Baddeley_spatial_2007} to statistically characterize the node positions over different network spatial realizations.
In fact, the use of these tools to model wireless networks started in the early 80's, when Takagi and Kleinrock firstly introduced such an approach to evaluate the aggregate interference power for Poisson distributed interferering nodes \cite{Takagi1984}.
This idea has been further developed and we can cite \cite{Baccelli2009_1,Baccelli2009_2,Haenggi,Cardieri2010} as basic tutorials on the topic.
Another important result which incorporates both the node positions' characterization based on stochastic geometry and the information-theoretic concept of capacity regions using different decoding rules was presented by Baccelli et al. in \cite{Baccelli2011}.
In that paper, the authors derived the capacity regions of Gaussian point-to-point codes for interference networks and then apply them to Poisson distributed networks.
An extension of that work have been recently proposed by the authors in \cite{Nardelli2012_WCNC_Spatial} so as to characterize the \emph{spatial capacity} of Poisson interference networks considering only the \emph{interference-as-noise} decoding rule.

In those works, though, neither queue analysis nor packet arrival process related to each node are considered, which can hide unstable scenarios.
In communication network theory, the queue stability has been extensively study for more general classes of stochastic networks (refer to \cite{Neely2010}) as well as in basic slotted ALOHA systems \cite{Luo1999}.
In \cite{Stamatiou2010}, Stamatiou and Haenggi started combining the stochastic geometric framework of modeling node positions and queuing theory used to assess the dynamics of node buffers to determine the stability region and average delay of single-hop ad hoc networks. 
This work was further extended in \cite{Nardelli2012_WCNC_Stability,Nardelli2012_submitted_Marios}, where the average spatial throughput is optimized under stability and packet loss  constraints such that the access probability, the number of possible retransmissions and the coding rate are the variables to be jointly tuned.
Other examples that include arrival processes into the spatial analysis, but more focused on comparing different access protocols that allows for asynchronous transmissions, are presented in \cite{kaynia_TWC_2010,Nardelli2012_accepted_Mariam}. 
%

%****************************************************************************
\subsection {Complex systems}
\label{SubSec_ReviewComplex}
Before we start talking about different problems that complex system analyses are employed, we think interesting to present the following quotation from \cite{Page_2010}: ``A complex system consists of diverse entities that interact in a network or contact structure -- a geographic space, a computer network, or a market. 
These entities' actions are interdependent -- what one protein, ant, person, or nation does materially affects others.
In navigating within a complex system, entities follow rules, by which I mean prescriptions for certain behaviors in particular circumstances''.
Besides, the rules just mentioned can be either fixed (e.g. physical laws) or adaptive (e.g. social behaviors).

The first example we introduce here is the \emph{tragedy of the commons} problem, which was described in \cite{Hardin1968}.
This problem can be stated as a social-economic dilemma such that many independent and rational agents share a given pool of limited resources.
In this scenario, the agents optimize their own pay-offs in a selfish manner, i.e. find their global optimum regardless of the others. 
Consequently, if every single agent takes the same decision, the shared resource will fade away after some time.
This conundrum is very context-dependent; for example, both fishing in a lake and forest usage can be  viewed as a tragedy of the commons class of problem, but the solution applied for each case tends to differ as the internal constraints of each system are different.
For this reason, after more than forty years of its publication, how to cope with it is still an open issue as we can see in \cite{Ostrom2007}.  

Another relevant problem about the interplay between coordination and cooperation is the well-known \emph{prisoner dilemma} \cite{Leyton-Brown2008}, in which two rational agents that cannot communicate to each other should choose whether to cooperate or not.
If both cooperate, they get a higher pay-off than if both do not cooperate.
However, if one cooperate and other does not, the non-cooperative agent will obtain a higher pay-off.
This fact leads to both agents not cooperating, which in turn provides lower pay-offs.
Several different studies based on this problem have been proposed under different assumptions and we will not discuss them here but one very interesting work recently proposed by Nowak \cite{Nowak2012}, where the author describe different ways that cooperative behavior can emerge in evolutionary, specially biological and social, systems.

Formation of cultures can be also viewed as a complex systems.
For example, the authors in \cite{Bednar2010} proposed an agent-based modeling to explain the existence of different set of behaviors within and across different populations.
The main idea is that two forces act in each individual, namely internal desire of \emph{consistence} and the social pressure of \emph{conformity}, and these forces will build the attributes of the collection of individuals, defining then their so-called culture.
The reader can refer to \cite{Page_2007} to explore more social models using the complex system approach.

Similarly the field of linguistics is stepping towards a complex understanding of the dynamics of languages.
In \cite{Kretzschmar_2007}, the author proposes an unified vision of speech theory where the language uses are analyzed as a spatial-temporal evolution of a complex system.
From this perspective, some intriguing phenomena such as the language variation used by a person within different groups and the geographical variety of languages can be better understood and explained.
This new view put some light on the role of regulative grammars as well as on the diversity that exists within languages.

In \cite{Jones-Rooy2010}, Jones-Rooy and Page offer a complex system analysis of the \emph{global systems history}.
For us, even more important than the specific focus of such an essay are the arguments used therein, which, we believe, are able provide a very instructive guidance on how the complex system thinking should be applied to visualize and model general phenomena.
Due to space limitation, we will not extend this survey to other important related topics as, for instance, network structures \cite{Barabasi2002,Newman2006} or exploration--exploitation trade-offs \cite{Page_2010}; yet we suggest the readers to go through the aforementioned essay, where such examples and some others are covered in a didactic manner.

To conclude this section, we want to say some few words about the seminal work \cite{Haykin2005}.
In that paper, Haykin stated the main features of the so-called \emph{cognitive radio}, which in his own words  is ``(...) defined as an intelligent wireless communication system that is aware of its environment and uses the methodology of understanding-by-building to learn from the environment and adapt to statistical variations in the input stimuli, with two primary objectives in mind: highly reliable communication whenever and wherever needed; efficient utilization of the radio spectrum''.
This work indicates the direction to more efficient wireless systems, whose designing is clearly related to the fundamentals of complex systems.
In fact, a huge number of papers that apply this concept have been published and such an idea is fairly well established in the academic community as well as in industries and operators.

We now return to our initial argument that, even though the research by-product of the cognitive radio concept has been facing a constant development and many aspects of complex systems have been already addressed, we believe that a deeper theoretical understanding of interference networks and their entities still lack, and this will be the target of this paper.
Our expectation is to indicate a new, unusual way of seeing wireless ad hoc networks based on how relation rules should be implemented and adapted from the available and locally processed information so as to guarantee a more robust network performance in relation to both internal and external pressures' variations.
%

%****************************************************************************
%****************************************************************************
\section{General Scenario Description and Formal Definitions}
\label{sec_general_description}

In this section we provide the theoretical background necessary to state the main claims of this paper.
Firstly we introduce the coding-decoding scheme employed in the scenarios analyzed here, followed by the description of other strategies, namely  access control mechanisms and packet retransmissions.
Then, we state the metrics applied to evaluate each link and the overall network performance in a given period of time.
%

%****************************************************************************
\subsection{Coding-decoding scheme}
\label{subsec_coding}
Let us start assuming an interference network composed by $K+1$ single-hop source-destination pairs (also called transmitter-receiver pairs) distributed over an given area of $A$ [m$^2$].
For this scenario, we revisit the basic statements of \cite[Sec. II]{Baccelli2011} that characterize the capacity region of Gaussian point-to-point (G-ptp) codes for an arbitrary number of communication pairs distributed over a given area.
We consider that each source node $i \in [0,\; K]$ wants to transmit an independent message $M_i \in \left[1:\; 2^{n R_i}\right]$, where $n$ is the block code length, to its respective destination $i$ at rate $R_i$ [bits/s/Hz].
Let $\mathbf{X} = (X_0,X_1,...,X_K)$ denote the set of transmitted signals  and $Z_i \sim \mathcal{CN}(0,1)$ be the complex circularly symmetric Gaussian random variable that represents the noise effect, then the received signal $Y_i$  at receiver $i$ is
\begin{equation}
\label{eq_detected_signal}
    Y_i = \sum_{j=0}^K g_{ij} X_j + Z_i,
\end{equation}
where $g_{ij}$ are the complex channel gains between transmitter $j$ (TX$_j$) and receiver $i$ (RX$_i$).
Then, considering that every transmitted signal is subject to the same power constrain of $Q$ [W/Hz], the received power  at RX$_i$ related to TX$_j$ is given by $P_{ij} = |g_{ij}|^2 Q$. 
Now, we assume that each transmitter (TX) uses a G-ptp code with a set of randomly and independently generated codewords $x_i^n(m_i) = (x_{i1},...,x_{in})(m_i)$ following i.i.d. $\mathcal{CN}(0,\sigma^2)$ sequences such that $0< \sigma^2 \leq Q$, where $m_i \in \left[1:\; 2^{n R_i}\right], \; i \in [0,\; K]$.
In the receiver (RX) side, a signal $y_i^n$ is received over the interference channel given by \eqref{eq_detected_signal} and an estimation $\hat{m_i}(y_i^n) \in \left[1:\; 2^{n R_i}\right]$ of the transmitted message can be then obtained.
An error event in the decoding happens whenever the transmitted message is not the estimated one.
Using this fact, we can state the error probability of our G-ptp code as follows:
\begin{equation}
\label{eq_error_probability}
    p_n = \dfrac{1}{1+K} \; \sum_{i=0}^K \textup{Pr}[\hat{M_i} \neq M_i].
\end{equation}

Based on the above, we define next the achievable rates and the capacity region for G-ptp codes.

\begin{definition}[achievable rates and capacity region]
\emph{Let $\overline{p}_n$ be the average of the error probability $p_n$ of a G-ptp code.
Then, a rate tuple $\mathbf{R} = (R_0,...,R_K)$ is said to be achievable if $\overline{p}_n \rightarrow 0$ when $n \rightarrow \infty$.
In addition, the capacity region using G-ptp codes is the closure of the set of achievable tuple rates $\mathbf{R}$.}
\end{definition}
 
Next, this definition is applied to establish which are the conditions for a rate tuple be achievable, obtaining afterwards the network capacity region.
 
\begin{theorem}[capacity region rewritten from \cite{Baccelli2011}]
\label{the_cap_reg}
Let $\mathbb{A}_i$ denote a subset of all $K+1$ transmitters that contains the TX$_i$ with $i  \in [0,\; K]$  and $\bar{\mathbb{A}}_i$ its complement.
Each receiver $i$  observes therefore a multiple access channel whose capacity region $\mathcal{C}_i$ is computed as
\begin{equation}
\label{eq_MAC_capacity_region}
    \mathcal{C}_i = \left\{ \mathbf{R}: \; \underset{k \in \mathbb{A}_i}{\sum} R_k \leq \log_2\left( 1 + \dfrac{\underset{k \in \mathbb{A}_i}{\sum} P_{ik}}{1 + \underset{j \in \bar{\mathbb{A}}_i}{\sum} P_{ij}} \right) \forall \;\; \mathbb{A}_i \subseteq \mathbb{A} \right\},
\end{equation}
and then the capacity region $\mathcal{C} $ of the Gaussian interference channel with G-ptp codes is obtained as
\begin{equation}
\label{eq_capacity_region}
    \mathcal{C} =  \bigcap_{i=0}^{K} \mathcal{C}_i.  
\end{equation}
\end{theorem}

As one can notice, the capacity region of Gaussian interference channel using G-ptp codes stated above requires a decoder that treats some of the interfering signals as noise, while others have their messages jointly decoded with the desired one.
This result indicates how the the capacity-achieving decoding strategy should be designed and hereafter we denote this receiver as OPT.
Such a solution, however, relies on the knowledge of the interfering TXs coding books and it can be computationally hard if many messages are jointly decoded.
For this reason, we also consider a simpler decoding rule where all interfering TXs are treated as noise and this decoding rule is denoted IAN.

From above, we can state the achievable rates under the IAN and the OPT decoding rules as follows.

\begin{theorem}[achievable rates for IAN decoding rule]
\label{the_IAN_achi}
Assuming that the noise is Gaussian and the TXs employ G-ptp codes, then the rate $R_k$ associated with a given link  TX$_k$-RX$_k$ is said to be achievable when the IAN strategy is used if, and only if, the following inequality holds:
\begin{equation}
\label{eq_IAN_inequality}
    R_k \leq \log_2\left( 1 + \dfrac{P_{kk}}{1 + \underset{j \in \mathbb{A} \backslash \{ k \} }{\sum} P_{kj}} \right),
    \vspace{-1ex}
\end{equation}
where $ \mathbb{A}$ represents the set of active transmitters.
\end{theorem}

\begin{proof}
This is a well-known result from information theory and can be viewed as a special case of \eqref{eq_MAC_capacity_region}, assuming that only the message of TX$_k$ is decoded by RX$_k$ while the others are treated as noise.
\end{proof}
\begin{theorem}[achievable rates for OPT decoding rule]
\label{the_OPT_achi}
\emph{Assuming that the noise is Gaussian and the TXs employ G-ptp codes, then the rate $R_k$ associated with a given link  TX$_k$-RX$_k$ is said to be achievable when the OPT decoder is employed if, and only if, the following inequality holds:
\begin{equation}
\label{eq_OPT_inequality}
      R_k \leq \log_2\left( 1 + \dfrac{\underset{i \in \mathbb{A}^*_k}{\sum} P_{ki}}{1 + \underset{j \in \bar{\mathbb{A}}^*_k}{\sum} P_{kj}} \right) - \underset{i \in \mathbb{A}^*_k \backslash \{ k \}}{\sum} R_i,
\end{equation}
where $\mathbb{A}^*_k$ represents the subset of transmitters whose messages are decoded by receiver $k$ and $\mathbb{A}^*_k \cup \bar{\mathbb{A}}^*_k = \mathbb{A}$ is the set of all active transmitters throughout the network.
}
\end{theorem}

\begin{proof}
To obtain \eqref{eq_OPT_inequality}, we proceed with a simple manipulation of equation \eqref{eq_MAC_capacity_region} in order to isolate the rate $R_k$ related to  TX$_k$-RX$_k$ link, considering the subsets $\mathbb{A}^*_k$ that lead to achievable rates.
\end{proof}

From this, we can define an outage event as the following.

\begin{definition}[outage event]
\label{def_out}
\emph{A given link  TX$_k$-RX$_k$ is said to be in outage if the coding rate $R_k$ for the decoding rule used is not achievable during any period of the message transmission.}
\end{definition}

%****************************************************************************
\subsection{Access control and retransmission  strategies}
\label{subsec_access_retx}
We describe here two different ways to access the channel in a random manner, namely slotted ALOHA and Carrier Sensing Multiple Access (CSMA), as studied in \cite{kaynia_TWC_2010,Nardelli2012_accepted_Mariam}. 
The advantage of these schemes is that they work without any central controlling, except by synchronization feature in the case of slotted ALOHA protocol that can be easily implemented if the nodes share the same internal clock.
In addition, we indicate a time-division scheme that, differently from random access protocol, determines the groups of concurrent transmissions for each given time-slot.
This strategy, though, can be implemented in a distributed fashion, maintaining the advantages of its random counterparts. 

In the slotted ALOHA protocol, each node TX$_k$ attempts to send messages or information packets to its respective RX with a given probability $p_k$ at the beginning of each time-slot.
Otherwise, with probability $1-p_k$, TX$_k$ stays in silence and thus does not cause any interference in that time-slot.
In the CSMA, on the other hand, there is no such synchronized communications and the TXs try to access the network whenever the channel is considering idle. 
The key mechanism of CSMA is the so-called \emph{carrier sensing} which is performed before each transmission attempt.
If the channel is considered idle by the decision-making node (which can be either/both TX or/and RX), the transmission begins immediately. 
When the channel is considered busy, the transmission is then backed off for a period of time period (normally random).
The back-off procedure may be repeated as many times as needed until the node finds the channel free or limited by a maximum number of attempts.

Besides, time-division schemes can be implemented in a distributed fashion if synchronism is possible.
Under this scheme, TXs will be divided into subsets that will transmit at the same time (e.g. some TXs will transmit in odd time-slots, while others in the even ones).
This will decrease the number of interfering TXs and improve the network condition. 
However, this  performance improvement is achieved by decreasing the effective time used to transmit (e.g. if we divide the TXs into two subsets, they will have half of the time to transmit their messages, decreasing their performance).

All these access procedures, however, do not avoid outage events as stated by Definition \ref{def_out} and decoding errors may occur.
Then, to avoid information losses, messages or packets that have not been successfully decoded during a given transmission attempt can be retransmitted.
Intuitively to implement a retransmission strategy, a feedback channel between the TX-RX links should be available to send retransmission requests.
As in the case of back-offs, retransmissions can be allowed until the message is successfully transmitted or up to a maximum number.

If this number is unbounded, packets can be retransmitted until a successful decoding and then the packet loss rate (PLR) tends to zero.
This PLR in turn comes at cost of a lower spectral efficiency since more channel usages may be necessary to successfully transmit the desired information.
Moreover, unbounded retransmission also increases the time in which the packet stays in the queue, which may create stability problems as we will discuss in the next paragraph.
On the other hand, when the number of retransmissions is limited by a given number, the PLR is not zero, but rather it is a function of such number as well as of other network variables.
In any case, if such a function is known, it is always possible to bound the PLR by a given small value.

As mentioned before, when many retransmissions and/or back-off are required to transmit successfully a message, the queue systems of the TXs start facing issues regarding their own stability.
To model this, let us assume here a single-server discrete-time queuing system\footnote{Even when continuous time systems are considered, this modeling can be applied if the arrival and server processes can be discretized accordingly.}, the backlog $Q_k(t)$ (queue length) of TX$_k$ with $k\in \mathbb{A}$ is determined for $t \in \left\{0, 1, 2, \ldots \right\}$ by \cite{Neely2010}:
\begin{equation}
\label{eq_queue1}
	Q_k(t+1)= \max[Q_k(t) - Y_k(t),0] + X_k(t), 
\end{equation}
where $\left\{Y_k(t)\right\}_{t=0}^{\infty}$ is the server process of  TX$_k$ and the initial queue lengths $\left\{Q_k(0)\right\}$ are chosen independently across TXs according to some probability distribution.

Based on such a equation, we can define the queue stability as follows \cite{Luo1999}.

\begin{definition}[stability]
\label{def_stability}
\textit{A backlog process $\left\{Q_k(t)\right\}$ is stable if the following holds
\begin{equation}
\label{eq_stable_cst}
\displaystyle \lim_{t \to \infty} \mathbb{P}[Q_k(t) < x] = F_{Q_k}(x) \ \ \ \ {\rm{and}} \ \ \ \ \lim_{x \to \infty} F_{Q_k}(x) = 1,
\end{equation}
for $x \in \mathbb{R}$, where $F_{Q_k}(x)$ is the limiting distribution function of $\left\{Q_k(t)\right\}$. }
\end{definition}
Clearly, the stability of the queues depends on both $\left\{X_k(t)\right\}_{t=0}^{\infty}$ and $\left\{Y_k(t)\right\}_{t=0}^{\infty}$.
While the former is an input parameter that the network elements cannot control, the latter is determined by the medium access protocol, the retransmission policy and the outage probability.
We can associate the queue system of TX$_k$ and its stability to an \emph{internal pressure} experienced by TX$_k$ to send information to RX$_k$, as we will see later.

\subsection{Network performance measures}
\label{subsec_metrics}
We present now two possible performance measures that can be used to evaluate the efficiency of the interference network studied here.
Specifically we consider one to assess the transmissions of each link and another to quantify the whole network efficiency.
\begin{definition}[effective link throughput]
\label{def_effthrough}
\emph{Let $R_k$ be the rate used by TX$_k$ to code its message to RX$_k$ that is transmitted during a given period of time $[t_1,t_2)$.
Then, the effective rate $\mathcal{R}_k$ measured  in [bits/s/Hz] is defined as
\vspace{-1ex}
\begin{equation}
\label{eq_link_rate}
    \mathcal{R}_k = R_k \cdot \Pr[\textup{Outage event at RX$_k$ during } [t_1,t_2)].
    \vspace{-1ex}
\end{equation}
}
\end{definition}

\begin{definition}[spatial throughput]
\label{def_spathrough}
\emph{Let $A$ [m$^2$] be the network area under analysis.
Then, during a given period of time $[t_1,t_2)$, the spatial rate $\mathcal{S}_{[t_1,t_2)}$ measured in [bits/s/Hz/m$^2$] is defined as
\vspace{-1ex}
\begin{equation}
\label{eq_spatial_rate}
    \mathcal{S}_{[t_1,t_2)} = \dfrac{1}{A} \; \sum_{i \in \mathbb{A}} \dfrac{\Delta t_i}{t_2-t_1}  \mathcal{R}_i,
    \vspace{-1ex}
\end{equation}
where $ 0 \leq \Delta t_i \leq t_2 - t_1$ is the length of time in which TX$_i$ is transmitting to RX$_i$ during the interval $[t_1,t_2)$ and recall that $\mathbb{A}$ is the set of all TXs in the network.
}
\end{definition}

We will show in the next section how these measures can be employed the network elements to  adaptively change its design setting so as to improve the system performance.

%****************************************************************************
%****************************************************************************
\section{A Complex System View of Interference Networks}
\label{sec_complview}
In this section we finally show how a complex system thinking can help to understand the limits of interference networks as well as the desired capabilities to reach such limits.
First we are going to identify which are the internal and external pressures experienced by the network elements as well as other external factors that may affect the system performance based on the modeling previously presented.
Then we will study two different extreme cases, namely quasi-static network and highly mobile network; for them, we will make claims about how the communication system should be designed to improve its efficiency mainly based on results presented in \cite{Baccelli2011,Nardelli2012_WCNC_Spatial,Nardelli2012_WCNC_Stability,Nardelli2012_submitted_Marios,Nardelli2012_accepted_Mariam}.
Finally, we apply those claims to make a guideline on the implementation of an adaptive algorithm employed by each node and tailed for interference networks, highlighting which capabilities are needed to build it.

\subsection{Pressures and external factors}
\label{subsec_press}
Let us start by defining internal pressure, external pressure and external factor for the interference network under analysis as follows.
Internal pressures are the constraints that each network element has to satisfy due to its own quality requirement.
In our case, we can cite as example of possible internal pressures a minimum coding rate (Section \ref{SubSec_ReviewCapacity}), a bounded outage probability (Definition \ref{def_out}), a  minimum required effective throughput (Definition \ref{def_effthrough}), a maximum PLR after back-offs and retransmissions (Section \ref{subsec_access_retx}) and buffer stability (Definition \ref{def_stability}).

In a similar way we can relate the external pressures to the constraints imposed by the network.
For example, we can list the fairness of the medium access across different links (i.e. nodes should have similar opportunities of access the medium), a maximum transmit power used by TXs in order to control the interference level and a floor level of spatial throughput (Definition \ref{def_spathrough}).
It is important noticing that different from the internal pressures the items listed above are not controlled by the link or any other entity, but rather they are product of the interactions among links that in turn are subject to their own internal pressures.

Clearly, how to cope with such interactions while preserving the overall network requirements and at the same time satisfying the internal pressures of each individual link can be viewed as the biggest designing challenge that engineers should deal with.
To complicate even more this picture, the network should be robust enough to variations of external factors.
For example, noise level and traffic conditions can vary during the day, mobility of nodes causes changes in the network topology, power blackouts can occur and so on; even under these wide range of different, many times unpredictable, external conditions, the system should work properly.

We can make here a parallel with biological or social systems, where the solution of problems related to internal and external pressures as well as variation of external conditions might emerge in a self-organized manner.
In such scenarios, however, nothing can be guaranteed and long time can be experienced to achieve a satisfactory solution, which in turn might be neither robust nor stable (e.g. \cite{Mitchell2009,Nowak2012}); even worse, sometimes solutions can vary a lot and this is not desired by engineering systems.

To build a robust adaptive communication system for interference networks that does not face such problems, we present next the basis of a practical algorithm that is able sustain adaptive and self-organized behaviors, making the network functional under different conditions. 
More specifically we will describe two scenarios where the conditions of mobility are extreme - highly mobile and quasi-static topologies - and then we will make claims regarding how the network should be designed to improve its performance under each one of those conditions.
Based on those claims, we will then discuss the construction of a stand alone procedure that is functional for interference networks and works similarly to the solution offered by intelligent humans in the chatting in a party problem presented in Section \ref{SubSec_Informal}.

\subsection{Cases of study}
\label{subsec_cases}
In this section we apply the modeling and the terminology previously stated to assess two examples of decentralized network deployments, namely quasi-static and high-mobile.
For each one of these scenarios, we will make claims regarding the design setting based on both the network conditions and pressures on its elements. 
Here it is worth pointing out that we are not going to present formal proofs of the claims, which in fact are conclusions that one can be drawn  from the following works available in the literature  \cite{Weber2012,Baccelli2011,Nardelli2012_WCNC_Spatial,Nardelli2012_WCNC_Stability,Nardelli2012_submitted_Marios,Nardelli2012_accepted_Mariam}.
%

%****************************************************************************
\subsubsection{Quasi-static network}
\label{subsec_quasistatic}
Let us assume here a network where the elements are quasi-static, which indicates its topology changes very slow during the day.
Examples of this can be an office with desktop computers, electro-electronics devices in a house or machines in an industrial plant.
For this scenario, we consider that the channel gains are strictly related to the distance-dependent path-loss \cite{yacoub_foundations_1993} (i.e. channel gains related to slow and fast fading are neglected) and that all TXs are subject to the same power constraint imposed by the network.
Based on these assumptions and considering that all TXs are able to estimate the distances between its own RX and the interfering TXs, we can make the following claims.

\begin{claim}[achievable rate]
\label{claim_achievable_quasi_static}
The TXs can individually code their messages such that the coding rate is achievable for that static topology regardless of the decoding rule used (refer to Theorems \ref{the_IAN_achi} and \ref{the_OPT_achi}). 
\end{claim} 

\begin{claim}[decoding rules]
\label{claim_dec_quasi_static}
If the network is sparsely populated, IAN and OPT decoders perform similarly in terms of spatial throughput.
Conversely, if the network is densely populated, OPT significantly outperforms IAN.
\end{claim} 

\textbf{Remark:} The advantage of OPT is obtained at expanse of a more computationally complex decoder. 
In addition, to jointly decode some messages, the RXs must know the codebook of other interfering TXs, which is not always feasible or desirable.

The claims and the remark stated above are mainly based on the results presented in \cite{Baccelli2011,Nardelli2012_WCNC_Spatial}.

\begin{claim}[access protocol and retransmissions] The design choices conditioned by the network conditions are stated below.
\label{claim_access_quasi_static}

\begin{itemize}
\item If all nodes are able to transmit with achievable rates (i.e. the network operates within its capacity region given by Theorem \ref{the_cap_reg}) and there is no minimum coding rate requirement, there is no need for retransmissions or medium access protocols.

\item If there exist a minimum required coding rate, then distributed time-division should be used to achieve such a rate.

\item If the network is densely populated, then distributed time-division schemes should be used to increase the spatial throughput.

\item If the traffic conditions are heavy and there is no minimum required coding rate, the TXs should find the best coding rate for the time-division employed such that each TX is able to maintain the stability of its own queue system. 

\item When a feasible combination does not exist for such heavy traffic conditions, the network will operate outside its capacity region. Therefore, distributed time-division schemes should be implemented together with a limited number of retransmissions, allowing for a bounded PLR.
\end{itemize}
\end{claim} 

All these claims provide us some intuition of the design setting that a quasi-static interference network should have under different conditions. 
It is important to mention that these claims hold only when every network element has the knowledge of: (i) the distances to the other TXs, (ii) its own traffic conditions and (iii) network density.
Therefore, all TXs must have the capability of \emph{sensing} the available signals to then compute \emph{estimations} of the network and traffic conditions.
With these estimations of the external factors in hand, the elements must find solutions that satisfy their own internal and external pressures such as buffer stability and transmit power constraint.
But before going further into implementation issues, which will be our focus later on Section \ref{subsec_imple}, we still need to analyze highly mobile topologies as presented next.
 
%****************************************************************************
\subsubsection{Highly mobile network}
\label{subsubsec_highmob}
Here we consider a network where its elements are highly mobile such that their positions change very fast so properties of Poisson point process can be applied using tools of stochastic geometry (refer to displacement theorem in \cite{Baccelli2009_1}).
We can see such a situation in shopping malls, streets, coffee houses or wherever place with intense flux of people using mobile devices.
We assume that channel gains are a composition of distance-dependent path-loss and fast fading \cite{yacoub_foundations_1993}.
As before, all TXs are subject to the same power constraint imposed by the network.
Then, we can make the following claims based on  \cite{Weber2012,Baccelli2011,Nardelli2012_WCNC_Stability,Nardelli2012_submitted_Marios,Nardelli2012_accepted_Mariam} considering that every TX knows the distance to its own RX as well as the spatial density of interfering nodes.

\begin{claim}[achievable rate]
\label{claim_achievable_highmob}
Every TX can individually code its messages so that the link outage probability can be bounded.
\end{claim} 

\textbf{Remark:} Due to the high mobility of the nodes, the network is not able to work within its capacity region. Yet, it is possible to bound the link outage probability by setting a proper coding rate.

\begin{claim}[decoding rules]
\label{claim_dec_highmob}
The OPT decoder is not a feasible option for highly mobile networks, while IAN is still functional under such conditions. 
\end{claim} 

\textbf{Remark:} A necessary condition to the OPT decoding rule is the knowledge of the coding book of other TXs. 
As the network topology changes very fast, it is not feasible to have such a knowledge.
Moreover, when TXs move during the transmissions, the subset of messages that are jointly decoded and treated as noise may also change.
All in all, it is very difficult, or even impossible, to employ OPT decoders in highly mobile scenarios.

\begin{claim}[access protocol and retransmissions] The design choices conditioned by the network conditions are stated below.
\label{claim_access_highmobi}
\begin{itemize}

\item A limited number of retransmissions for packets detected in error improves the system performance regardless of the medium access procedure. 

\item Distributed time-division is not a feasible option for highly mobile networks.

\item If the network is sparsely populated, then random access protocols that employ carrier sensing (e.g. CSMA) tend to outperform slotted ALOHA schemes.

\item If the network is densely populated and the traffic conditions are not heavy, the synchronous transmissions required by the slotted ALOHA  improve both individual link and spatial throughputs.

\item If the traffic conditions are heavy and the network is densely populated, the TXs should find the best combination of access probability under slotted ALOHA protocol, coding rate and number of allowed retransmissions per packet such that every TX is able to maintain the stability of its own queue system and have a bounded PLR. 
\end{itemize}
\end{claim} 

Here once again one can see that every network element must be able to sense and estimate the network and traffic conditions, using them to assess the feasibility of possible design setting based on their internal and external pressures.
Using these claims and the ones proposed in the previous section, we will discuss in the following how a robust adaptive algorithm that allows for self-organization should be designed for interference networks under different conditions of density, traffic and mobility.

\subsection{Implementation discussion}
\label{subsec_imple}
In the previous subsections we stated several claims describing how  the design of interference networks should be for two extreme mobility scenarios based on theoretical results found in the literature.
Such works, however, lack of a deeper understanding of more practical deployments in which the network condition is dynamic, i.e. during some time the network can be high mobile while during some other periods it can be quasi-static or neither one.
We can visualize this condition in, for instance, smart homes where there are equipments communicating that are static and there might be people using other communication devices as well.
In there, sometimes there is no one around or people are sleeping and then only machines communicate.
In contrast, during other periods they are awake using their wireless devices, walking around and generating traffic.

Our aim here is to guide the design of a possible adaptive algorithm to improve the network efficiency while the constraints associated to internal and external pressures can be satisfied. 
In the following we indicate the basis of such an algorithm, which should be employed by all elements of the network.
%

%****************************************************************************
\subsubsection{Variable to be optimized}
In \cite{Nardelli2012_submitted_Marios} the authors indicated that if all links optimize their own effective throughput given by Definition \ref{def_effthrough} in a selfish way, then the common resource might be overused, leading to  similar effects to the tragedy of the commons problem \cite{Hardin1968}.
There, they also showed that if all links optimize the network spatial throughput given by Definition \ref{def_spathrough}, then they can reach link effective throughputs at least as high as in the selfish optimization case.
The reason of this is the following.
While the former procedure optimizes the link performance for a \emph{given} network state which is in turn considered independent of its own decision, the latter assumes that the designing choice of every TX affects the others and then also modifies the actual network state.
Therefore, by optimizing the network spatial throughput assuming that all links proceed in the same way, the setting that is the optimal for the network is also the optimal for each individual link.
Motivated by this result, we argue that the spatial throughput is the measure to be optimized by the most efficient algorithms designed for interference networks since it provides incentives to collaborative behavior, avoiding then prisoner dilemmas kind of loses \cite{Nowak2012}.

%****************************************************************************
\subsubsection{Variables required to proceed with the optimization}
Now that the maximization target is defined, we should determine what each network element should know to optimize its performance.
First of all, they should assess their own internal pressures as their basic constraints.
For example, every TX should infer its arrival process, which is an external factor, to determine the conditions that guarantee its own queue stability and therefore it will be able to determine the feasibility of possible solutions.

Then they should also estimate the mobility pattern of the network to evaluate how its topology changes\footnote{If the node is moving itself, it will see a mobile network around it and will decide to design its communication system accordingly.}.
Once the mobility pattern is identified, the nodes should quantify the network density and/or the distances from each other using procedures as the ones presented in \cite{Onur2012} and \cite{Yang2012}, respectively.
Clearly, it is important that all nodes  assume in their calculations that the external pressures suffered by the other nodes are the same (e.g. the same power constraint).
Given that such elements locally estimate those information about the network state, they can start optimizing their own design setting based on the claims previously stated.
In other words, given the internal and external pressures, the mobility pattern and the knowledge of the distances and/or the network spatial density, each node will set, for example, the coding rate employed, the decoding rule, the medium access procedure and the maximum number of packet retransmissions that jointly maximize the network spatial throughput, which is locally computed from the estimated network state, while all constraints are satisfied.
If the link starts facing problems in satisfying its own constraints or after a given period of time, the procedure should be repeated to adapt its setting to the new state of the network. 

As one can notice, this algorithm mimics how humans solve the chatting in a party problem defined in Section \ref{SubSec_Informal}.
It is also worth saying that we choose here to not go into the specificities of algorithms or signal processing schemes; rather we prefer to provide more general guidelines that apply a complex system view of the interference network problem.

%****************************************************************************
%****************************************************************************
\section{Final Discussion and Future Works}
\label{sec_Concl}
Throughout this paper, the problem of interference-limited wireless networks was revisited from a complex system perspective.
Specifically we identified how the performance of such a class of networks can be determined by internal and external pressures as well as by external factors.
From this characterization, we discussed in general lines the implementation of adaptive algorithms in order to design functional interference networks that might be subject to diverse conditions of mobility and traffic.

It is important to mention that implementing a system that actually employs  the ideas presented here requires the specific knowledge of the application scenario --  what is required by a smart home is different from what is required by a industrial plant and the algorithm should reflect it.
In other words, there is no panacea \cite{Ostrom2007}.
Nevertheless, one can verify that the capabilities needed by  different applications that build interference network structures are very similar, regardless of such specificities.
All elements of a functional interference network must be able to sense, store messages, process data and optimize mathematical functions.
In addition, every element should always be aware that its own action will affect and interfere with the overall network efficiency, going in similar lines to what Mills pointed out in his fundamental work in the field of sociology about the so-called \emph{sociological imagination} \cite{Mills1959}.

In this context, a straightforward future work is to study real applications of interference networks such as smart grids or smart homes and then apply the ideas presented in this paper to design actual adaptive solutions that employ a complex system thinking.

%****************************************************************************
%%Bibliography%%
\bibliographystyle{IEEEtran}
\bibliography{ref_abbrev}

\end{document}